\documentclass[article,11pt,onefignum,onetabnum]{siamart220329}

\usepackage{lipsum}
\usepackage{amsfonts,amsmath,amssymb}
\usepackage{graphicx}
\usepackage{epstopdf}
\usepackage{algorithmic}
\usepackage{mathrsfs}
\usepackage{dsfont} 
\ifpdf
  \DeclareGraphicsExtensions{.eps,.pdf,.png,.jpg}
\else
  \DeclareGraphicsExtensions{.eps}
\fi

\newsiamremark{Remark}{Remark}
\newsiamremark{Corollary}{Corollary}
\newsiamremark{hypothesis}{Hypothesis}
\crefname{hypothesis}{Hypothesis}{Hypotheses}
\newsiamthm{claim}{Claim}

\headers{Quantum subspace selection and stabilization}{}

\title{Exponential Selection and Feedback Stabilization of Invariant Subspaces of Quantum Trajectories\thanks{Submitted to the editors DATE.
\funding{This work is supported by the ANR project “Estimation et contrôle des systèmes quantiques ouverts" Q-COAST Projet ANR-19-CE48-0003, the ANR project QUACO ANR-17-CE40-0007,  the ANR project IGNITION ANR-21-CE47-0015, the ANR project “Quantum Trajectories" ANR-20-CE40-0024-01, and ANR project “Évolutions Stochastiques Quantiques" ESQuisses Projet-ANR-20-CE47-0014.}}}

\author{Nina H. Amini \thanks{Laboratoire des Signaux et Syst\`{e}mes (L2S), CNRS-CentraleSup\'{e}lec-Universit\'{e} Paris-Sud, Universit\'{e} Paris-Saclay, 3, Rue Joliot Curie, 91190, Gif-sur-Yvette, France. 
  (\email{nina.amini@centralesupelec.fr}).},  \and Maël Bompais \thanks{Laboratoire des Signaux et Syst\`{e}mes (L2S), CNRS-CentraleSup\'{e}lec-Universit\'{e} Paris-Sud, Universit\'{e} Paris-Saclay, 3, Rue Joliot Curie, 91190, Gif-sur-Yvette, France. 
  (\email{mael.bompais@centralesupelec.fr}).}, \and  Clément Pellegrini \thanks{Institut de Mathématiques, UMR5219, Université de Toulouse, CNRS, UPS, F-31062 Toulouse Cedex 9, France (\email{clement.pellegrini@math.univ-toulouse.fr}.)}}

\usepackage{amsopn}

\usepackage{graphicx} 
\usepackage{a4wide}
\usepackage{mathrsfs}
\usepackage[utf8]{inputenc}
\usepackage[english]{babel}
\usepackage{graphicx} 
\usepackage{mathrsfs}
\usepackage{stmaryrd}
\usepackage{color}
\definecolor{lightgray}{gray}{0.7}
\usepackage[toc,page]{appendix}
\usepackage{color}
\usepackage{dsfont}
\usepackage{extarrows}
\usepackage{enumerate}
\usepackage{amsmath}    
\usepackage{amssymb}
\usepackage{bbm}
\usepackage{graphicx}   
\usepackage{mathrsfs}
\usepackage{verbatim}   
\usepackage{subfigure}  
\usepackage{hyperref}   
\usepackage{braket}
\DeclareMathOperator*{\argmin}{argmin}
\DeclareMathOperator\supp{supp}
\allowdisplaybreaks

\newtheorem{assumption}{Assumption}[section]
\newtheorem{remark}{Remark}[section]
\newenvironment{customproof}[1][\proofname]{\proof[#1]}{\endproof}

\renewcommand{\H}{\mathcal{H}}
\newcommand{\tr}[1]{\operatorname{tr} \left( #1 \right)}
\newcommand{\E}[1]{\mathbb E \left[ #1 \right]}
\newcommand{\EE}[2]{\mathbb E \left[ #1 \middle| #2 \right]}
\newcommand{\rhon}{\rho_{n}}
\newcommand{\rhonplus}[1]{\rho_{n+#1}}
\newcommand{\rhok}{\rho_{k}}

\newcommand{\rhoalpha}{\rho^{(\alpha)}}
\newcommand{\rhobeta}{\rho^{(\beta)}}
\newcommand{\rhoalphainfty}{\rho^{(\alpha)}_\infty}
\newcommand{\rhobetainfty}{\rho^{(\beta)}_\infty}
\renewcommand{\o}{\left(}
\renewcommand{\c}{\right)}

\newcommand{\alphabar}{\bar\alpha}
\newcommand{\ubar}{\bar u}
\newcommand{\normalize}[1]{ \frac{#1}{\tr{#1}}}
\newcommand{\Malphabar}{M_{\alphabar}}
\newcommand{\Halphabar}{\H_{\alphabar}}
\renewcommand{\leq}{\leqslant}
\renewcommand{\geq}{\geqslant}
\newcommand{\ketbra}[1]{\ket{#1}\bra{#1}}
\newcommand{\M}{\mathscr M}

\date{April 2023}

\begin{document}

\maketitle

\begin{abstract}
    We show that quantum trajectories become exponentially fast supported by one of their minimal invariant subspaces. Exponential convergence is shown in expectation using Lyapunov techniques. The proof is based on an in-depth study of the identifiability of the probability distributions generated in the different subspaces. We furthermore introduce a feedback control strategy that allows for the targeted convergence towards a desired subspace. This convergence is also achieved at exponential speed.
\end{abstract}



\section{Introduction}
Controlling and manipulating open quantum systems, which are systems that interact with their environment \cite{davies1976quantum,breuer2002theory,gardiner2004quantum,haroche2006exploring,wiseman2009quantum}, 
plays a crucial role in the realization of various quantum technologies, including quantum computation, quantum communication, and quantum engineering. Quantum properties are highly susceptible to disturbances caused by the environment and operational actions.
Directly measuring open quantum systems has a disruptive effect on their evolution and hinders their manipulation. Consequently, most experiments employ the indirect measurement procedure, wherein the system interacts with a series of probes \cite{braginsky1995quantum, haroche2006exploring}. After each interaction between the system and a probe, a measurement is performed on the probe. The entanglement between the system and the probe enables inference of the system's evolution based on the information obtained during the measurement. This evolution is typically random and can be described using Markov chains or stochastic differential equations. These processes, known as quantum trajectories, have significantly advanced the understanding and study of various aspects of quantum mechanics \cite{davies1976quantum,belavkin1987non,belavkin1999measurement,wiseman2009quantum,carmichael2009open,barchielli2009quantum}.

Diverse methods can be employed to control quantum trajectories in order to achieve objectives such as state preparation, stabilization, and enhancing the accuracy of indirect measurements \cite{mabuchi2005principles, ahn2002continuous, armen2002adaptive,Ticozzi,Ticcoopen,Altafini,catana2014heisenberg}.
These approaches encompass time-dependent as well as state-based feedback control strategies. 
The latter are widely recognized for their ability to robustly control open quantum systems \cite{doherty1999feedback,van2005feedback,sayrin2011real,mirrahimi2007stabilizing}.


One of the models that has attracted significant attention and initiated numerous advancements is the quantum non-demolition (QND) measurement model. It was in particular implemented by Serge Haroche's group to manipulate and observe photons in cavities without destroying them, and has been later developed in several scenarios, from fundamental aspects of quantum collapse to practical applications where control theory has played a significant role \cite{sayrin2011real}. The key feature of this model is that, under some specific conditions, the system converges to a pointer state. More precisely, a pointer state is a particular state of an orthonormal basis prescribed by the physical model (for example, the Fock states for photons). 
Within this framework, precise results of large-time behavior, speed of convergence, stability, robustness, and control have been obtained \cite{BEnoistAHP,BenrardBauer,BenoistPelCMP,amini2013feedback,Fraasnondemol,Fraasnondemol2}. The fact that quantum trajectories converge to pointer states can be referred to as the selection of pointer states. A control strategy can be designed to render this selection deterministic and force, for example, the cavity to have a prescribed number of photons in long time \cite{sayrin2011real}. In continuous time, stabilizing exponentially eigenstates of the measurement operator in the context of QND measurements with state-based feedback and noise-assisted feedback have been respectively addressed in \cite{Liang2018j} and \cite{cardona2020exponential}. 

The particularity of QND measurements is that each pointer state corresponds to an invariant subspace of dimension one of the dynamics. More precisely, with and without measurements, if the system is initialized in a pointer state, it will remain in this state (this is the essence of the name non-demolition). The one-dimensional nature of the invariant spaces simplifies their investigation, facilitating the study of many aspects such as the identifiability of different pointer states, convergence speed, and stabilization. QND models can indeed be expressed as mixtures of i.i.d. models, as highlighted in \cite{BenrardBauer,BEnoistAHP,BenoistJPHYSA,BenoistPelCMP}. As a consequence, classical probability theory can be effectively utilized to achieve the convergence result.


When considering generic measurements, not exclusively QND, the convergence of quantum trajectories can be studied through the decomposition of the Hilbert space. In analogy with classical Markov chain, the Hilbert space describing the quantum system can be decomposed into recurrent and transient parts \cite{baumgartner2012structures,carbone2016irreducible}. 
The transient part is the region of the Hilbert space from which the dynamics escape over time. 
The recurrent part, which is the orthogonal complement of the transient part, can be further decomposed into minimal invariant subspaces. 
These minimal invariant subspaces have arbitrary dimensions and are in one-to-one correspondence with minimal invariant states of the quantum channel
\cite{carbone2016irreducible,baumgartner2012structures,fagnola2003transience} (as the recurrence classes of a finite state Markov chain which are the support of minimal invariant measures). Roughly speaking, quantum trajectories will select one of the minimal invariant subspaces as the amount of information obtained through indirect measurement grows, provided that these subspaces can be identified from one another.
Such a result has been addressed for discrete-time quantum trajectories in~\cite{BePelLin} and continuous-time quantum trajectories in~\cite{MR4318593}.

This article considers generic discrete-time measurement processes within the context where the Hilbert space consists solely of the recurrent part. Then the Hilbert space can be decomposed (possibly non-uniquely) into a direct sum of minimal invariant subspaces. Our aim is twofold. First, we show that the selection of minimal invariant subspace is exponentially fast in expectation. This improves all the known results in this domain, where convergence considerations have been addressed only in an almost sure sense \cite{MR4318593,BePelLin}. Second, motivated by the exponential convergence, we design a feedback control strategy that stabilizes a targeted minimal invariant subspace at exponential speed again. To our knowledge, this is the first result with state-based feedback stabilizing a chosen minimal invariant subspace in a natural framework. Other studies have previously been conducted using a switching of the environment \cite{grigoletto2021stabilization,liang2022switching}, and Markovian output feedback as well as state-based feedback to stabilize a one-dimensional target subspace \cite{ticozzi2008quantum,ticozzi2012stabilization}. 

Surprisingly, showing the exponential convergence in the case of indirect measurement and in the control setup does not require strengthening the identifiability condition of \cite{BePelLin}. It is important to note that our goal is to achieve exponential convergence in mean. This, in principle, requires a detailed understanding of the dynamics for all states of the system, whereas identifiability only concerns the invariant states. A part of our contribution is to show that the identifiability of the invariant states implies the identifiability of all states supported in any two different invariant subspaces. More precisely, we show that we have a uniform identifiability, meaning we need only the probability of a finite sequence of measurement results to identify the invariant subspaces. To address the feedback stabilization of a chosen invariant subspace, we develop a feedback control strategy that consists of acting on the system at regular intervals. Our approach draws inspiration from prior research conducted on discrete time quantum trajectories in the context of QND measurements, for instance, on the LKB photon box \cite{mirrahimi2009feedback,amini2013feedback,amini2011design}. 
 The aforementioned studies solely focus on convergence, disregarding considerations related to the stabilization rate over time. Consequently, our present work surpasses previous findings by demonstrating an exponential stabilization rate and extending its applicability to arbitrary invariant subspaces, not restricted to QND measurements.

The paper is structured as follows. In Section \ref{sec:Prese}, we present the model of quantum trajectories, the minimal invariant space decomposition, and the identifiability assumption. This section provides the essential ingredients for showing exponential convergence.
In section \ref{sec:openloop}, we show the exponential selection of a minimal invariant subspace. In Section \ref{sec:feedback}, we present the exponential stabilization of a targeted minimal invariant subspace using feedback control. As will be seen throughout the article, our approach is based on the use of Lyapunov functions.
\section{Quantum trajectories and identifiability of probability measures}\label{sec:Prese}

This section is devoted to the presentation of the quantum trajectory model. In particular, we shall make precise the assumptions that will be used throughout the paper, namely the decomposition of the Hilbert space with no transient part and the identifiability. 

We consider a quantum system undergoing discrete-time indirect measurement. The Hilbert space of the system is denoted by $\H \cong\mathbb  C^d$ and the state space is the set of density matrices
$$\mathcal D=\{ \rho \in \mathcal B(\mathcal H) \mid \, \rho=\rho^\dag,\,\rho \geq 0,\, \tr{\rho}=1\},$$ 
where $ \mathcal{B}(\mathcal{H})$  is the set of linear operators on  $\mathcal{H}$. With an ideal detector, the state of the system evolving over time is a Markov chain $(\rho_n)$ defined by the transitions
\begin{equation}\label{traj}
    {\rho}_{n+1}=\frac{V_{i_n}{\rho}_{n}V_{i_n}^\dag}{\operatorname{tr}\left(V_{i_n}{\rho}_{n}V_{i_n}^\dag\right)},
\end{equation}
depending on the random measurement result $i_n$ taking values in a finite set $\mathcal O.$ The probability to detect the measurement result $i_n=i$ at time $n$ knowing that the state is $\rho_n$ just before the measurement is given by $\operatorname{tr} (V_{i}{\rho}_{n}V_{i}^\dag ).$ The operators $V_i$ are referred to as Kraus operators and satisfy $\sum_{i \in \mathcal O} V_i^\dag V_i=I_d,$ with $I_d$ the identity operator on $\H.$ 
In this paper, we shall consider only perfect detection, but everything can be easily generalized to imperfect detection models.\\

The Kraus operators $V_i,~i\in\mathcal O,$ define a quantum channel 
$\Phi$ by
$$\Phi(X)=\sum_{i\in\mathcal O}V_i XV_i^\dag$$
for all $X\in\mathcal B(\mathcal H)$. The map $\Phi$ is completely positive and trace-preserving. Its adjoint is denoted by $\Phi^\dag$. The quantum channel and the quantum trajectory are intimately linked. Indeed,
$$\mathbb E[\rho_{n+1}\vert\rho_n]=\sum_{i\in\mathcal O}\frac{V_i\rhon V_i^\dag}{\tr{V_i\rhon V_i^\dag}}\tr{V_i\rhon V_i^\dag}=\Phi(\rhon)$$
and then for all $n\in\mathbb N$, $\mathbb E[\rho_n]=\Phi^n(\rho)$ if $\rho$ is the initial state. For $n\in\mathbb N$, the set $\mathcal O^n$ represents the words that can be recorded during $n$ consecutive measurements. We introduce the probability measure $\mathbb P_\rho$ defined for all $n$ and all $I=\{i_1,\ldots,i_n\}\in\mathcal O^n$ by
\begin{equation*}\label{def:prob}\mathbb P_\rho \o \{i_1,\ldots,i_n\} \c =\tr{V_I\rho V_I^\dag},
\end{equation*}
where $V_I=V_{i_n}\ldots V_{i_1}.$ The probability measure $\mathbb P_\rho$ is naturally extended on $\mathcal O^{\mathbb N}$ endowed with the cylinder $\sigma-$algebra (using Kolmogorov consistency criterion, see \cite{benoist2019invariant,BePelLin}). We will not dwell on the details of this construction and just refer to the definition of $\mathbb P_\rho$ for words of finite size. Starting from an initial state $\rho$, if we observe $I=\{i_1,\ldots,i_n\}$ during the first $n$ measurements - this happens with probability $\mathbb P_{\rho}(I)$ -, then the quantum trajectory state at time $n$ is
$$\rho_n=\frac{V_I\rho V_I^\dag}{\tr{V_I\rho V_I^\dag}}=\frac{V_{i_n}\ldots V_{i_1} \rho V_{i_1}^\dag\ldots V_{i_n}^\dag}{\tr{V_{i_n}\ldots V_{i_1} \rho V_{i_1}^\dag\ldots V_{i_n}^\dag}}.$$
A particular useful fact that we shall use extensively is
$$\Phi^n(\rho)=\sum_{I\in\mathcal O^n}V_I\rho V_I^\dag.$$
This indeed implies that for all $n,\,p$, we have
$$\mathbb E[\rho_{n+p}\vert\rho_n]=\sum_{I\in\mathcal O^p}\frac{V_I\rhon V_I^\dag}{\tr{V_I\rhon V_I^\dag}}\tr{V_I\rhon V_I^\dag}=\Phi^p(\rhon).$$

Now let us discuss the invariant states of $\Phi$. To this end we introduce
$$\mathcal F_{\Phi}=\{X\in\mathcal B(\mathcal H) \mid \Phi(X)=X\}$$
for the set of fixed points. The set of invariant states is the convex set given by $\mathcal D_\Phi=\mathcal D\cap \mathcal F_\Phi$. Its extreme points are called minimal invariant states. 

Quantum channels are usually compared with transition operators of classical Markov chains. In particular, transient and recurrent parts for quantum channels can be defined in analogy with Markov chains \cite{BePelLin,fagnola2003transience,Bardet,MR3626627,carbone2016irreducible}. In our context, the transient part is defined as follows \cite{baumgartner2012structures}:
$$\mathcal T=\{x\in\mathcal H\mid\langle x,\Phi^n(\rho)x\rangle\xrightarrow[n\to\infty]{}0\},$$
and the recurrent part as  $\mathcal R=\mathcal T^\perp$. The set $\mathcal R$ supports the invariant states of $\Phi$, meaning that for any $\rho\in \mathcal D_\Phi$, $\supp{\rho}\subset \mathcal R$. In the paper, we concentrate on the case $\mathcal T=\{0\}.$ Then the Hilbert space $\mathcal H =\mathcal R$ accepts an orthogonal direct sum decomposition $$\mathcal H=\bigoplus_{\alpha=1}^\ell \H_\alpha,$$
where each $\mathcal H_\alpha$ is the support of a minimal invariant state. The subspaces $\H_\alpha$ are minimal enclosures for $\Phi$ (see \cite{baumgartner2012structures,carbone2016irreducible}), the term enclosure meaning that if $\supp\rho \subset \mathcal H_\alpha,$ then $\supp \Phi(\rho) \subset\mathcal H_\alpha,$ while the term minimal means that any subspace strictly included in $\H_\alpha$ does not verify this property. In the sequel, we shall denote $d_\alpha$ the dimension of $\mathcal H_\alpha$, and $\rho_{\infty}^{(\alpha)}\in\mathcal D_\Phi$ the minimal invariant state with support $\mathcal H_\alpha$. The minimality can also be understood as the fact that no other invariant state of $\Phi$ has its support strictly included in $\mathcal H_\alpha$. This decomposition is in general not unique (see \cite{baumgartner2012structures,carbone2016irreducible}), and we pick an arbitrarily fixed one.

For each $\rho_{\infty}^{(\alpha)}$, we consider the positive associated operator $M_\alpha$ satisfying $ \Phi^\dag(M_\alpha)=M_\alpha$ and $\operatorname{tr}(M_\alpha\rho^{(\beta)}_\infty)=\delta_{\alpha,\beta}$. In the case $\mathcal T=\{0\}$, each operator $M_\alpha$ is the orthogonal projector on $\H_{\alpha}$ and their sum satisfy $\sum_{\alpha=1}^\ell M_\alpha=I_d.$ One can show that the Kraus operators are block diagonal with respect to the decomposition $\bigoplus_{\alpha=1}^\ell \H_\alpha$ (see e.g., \cite{carbone2021absorption,Wo12}). In particular, this implies that if the initial state of the quantum trajectory is supported in $\mathcal H_\alpha$, then the quantum trajectory remains supported in $\mathcal H_\alpha$ at all time. 
These properties are equivalent to say that for all $i\in\mathcal O$ and all $\alpha,$ 
$M_\alpha V_i=V_iM_\alpha.$



At first glance, the reduction to the case $\mathcal T=\{0\}$ can be considered restrictive, but the results are already remarkable in this situation. The general case $\mathcal T\neq\{0\}$ is trickier, and \cite{BePelLin} provides information on this aspect, but does not touch upon the topics of exponential convergence and control.
Furthermore, our context is a generalization of the QND experiment. In the QND setting, all the invariant minimal subspaces are of dimension one, that is
$$\mathcal H=\bigoplus_{\alpha=1}^\ell \mathbb C\vert e_\alpha\rangle,$$
where $\{\vert e_\alpha\rangle,~\alpha=1,\ldots,\ell\}$ is an orthonormal basis usually called pointer basis, as mentioned in the introduction.

In \cite{BePelLin} for discrete-time quantum trajectories and in \cite{MR4318593} for continuous-time quantum trajectories, it is shown that quantum trajectories converge towards a minimal invariant subspace $\mathcal H_\alpha$ for some $\alpha\in\{1,\ldots,\ell\}$. This requires an identifiability assumption, expressed as follows in the discrete case.
\begin{assumption}[\textbf{ID}]
For any $(\alpha,\beta)$ with $\alpha\neq \beta$, there exists a sequence $I=(i_1,\ldots, i_p)$ such that
$$\tr{V_I \rho_{\infty}^{(\alpha)} V_I^\dag} \neq \tr{V_I \rho_{\infty}^{(\beta)} V_I^\dag}.$$
\end{assumption}
The next proposition establishes the selection of a minimal invariant subspace.
\begin{proposition}\label{prop:convergence}
    Let $\rho\in\mathcal D$ and let $(\rho_n)$ be the quantum trajectory with initial state $\rho$. Under Assumption \textbf{(ID)}, there exists a random variable $\Upsilon$ valued in $\{1,\ldots,\ell\}$ such that
    $$\mathbb P[\Upsilon=\alpha]=\tr{M_\alpha\rho}$$
    for all $\alpha \in \{1,\ldots,\ell \}$ and
    $$\lim_{n\to\infty}\tr{M_{\Upsilon}\rho_n}=1\quad a.s.$$
\end{proposition}
In other words, the quantum trajectory becomes asymptotically fully supported in the subspace $\H_\Upsilon$, randomly selected among the minimal invariant subspaces $\H_\alpha.$
Assumption \textbf{(ID)} is natural and cannot be relaxed. Indeed consider that for some $\alpha\neq\beta$ we have, for all $I$, $\operatorname{tr}(V_I\rho_{\infty}^{(\alpha)}V_I^\dag )=\operatorname{tr} (V_I\rho_{\infty}^{(\beta)}V_I^\dag).$ If the initial state is a non-trivial convex combination $\rho=\lambda\rho_{\infty}^{(\alpha)}+\mu\rho_{\infty}^{(\beta)},$ then the quantum trajectory will never select one of the subspaces. More precisely, a straightforward computation yields $\tr{M_\alpha\rhon}=\lambda$ and $\tr{M_\beta\rhon}=\mu$ for all $n\in\mathbb N$ and Proposition \ref{prop:convergence} fails.

The rest of this section is devoted to studying assumption \textbf{(ID)}. In the sequel, we shall use the notation $\rho^{(\alpha)}$ to denote a state whose support is in $\mathcal H_\alpha$.  In quantum filtering theory, thinking in terms of identifiability refers to the identifiability of probability measure \cite{VanHandelobserv}. This is the case in our context. First, assumption \textbf{(ID)} can be rephrased by saying that for any $\alpha\neq\beta$
$$\mathbb P_{\rho_{\infty}^{(\alpha)}}\neq \mathbb P_{\rho_{\infty}^{(\beta)}}.$$ At this stage, one can wonder: if $\alpha\neq\beta,$ do we have
$$\mathbb P_{\rho^{(\alpha)}}\neq \mathbb P_{\rho^{(\beta)}}$$
for all $\rho^{(\alpha)}$ and $\rho^{(\beta)}?$ This is indeed true. By contradiction, if it was not the case, it would mean that there exist $\rho^{(\alpha)}$ and $\rho^{(\beta)}$ such that $\mathbb P_{\rho^{(\alpha)}}=\mathbb P_{\rho^{(\beta)}}$. Then for all $k$, for all $J\in\mathcal O^k$ we have
$$\tr{V_IV_J\rho^{(\alpha)}V_J^\dag V_I^\dag}=\tr{V_IV_J\rho^{(\beta)}V_J^\dag V_I^\dag},$$
where $I$ is the word in assumption \textbf{(ID)} distinguishing $\rho_{\infty}^{(\alpha)}$ and $\rho_{\infty}^{(\beta)}$. Then summing over $J$, we get
\begin{equation}\label{eq:equality}
    \sum_{J\in\mathcal O^k}\tr{V_IV_J\rho^{(\alpha)}V_J^\dag V_I^\dag}=\tr{V_I\Phi^k(\rho^{(\alpha)})V_I^\dag}=\tr{V_I\Phi^k(\rho^{(\beta)})V_I^\dag}.
\end{equation}
At this stage, it is tempting to take $k\to\infty$. However, $\Phi^k$ does not converge in general, and we shall use \cite[Theorem 3.3]{benoist2019invariant} reformulated here, which is a corollary of Evans and Hoegh-Krohn Perron-Frobenius Theorem for positive maps \cite{evans1977spectral}.

\begin{theorem}\label{thm:ergod}
	Assume there exists a unique $\Phi$-invariant element $\rho_{\infty} \in \mathcal D_\Phi$. Then there exist $m\in\mathbb N$, two positive constants $c$ and $\lambda<1$ such that, for any $\rho\in\mathcal D$, for any $n\in\mathbb N^*$, and for all $s\in \{ 0,\ldots,m-1 \},$
	\begin{equation*}\label{eq:ergod}
	\Big\|\frac1m\sum_{r=0}^{m-1} \Phi^{mn+s+r}(\rho)-\rho_{\infty}\Big\|_1\leq c\lambda^n.
	\end{equation*}
 The minimal integer $m$ is usually called the period of the quantum channel $\Phi$. The case $m=1$ is called aperiodic case, while the periodic case refers to $m>1$.
\end{theorem}
 In our setting where $\Phi$ may have multiple invariant elements, we shall use the following corollary.

\begin{corollary}
    There exist $m\in\mathbb N$, two positive constants $c$ and $\lambda<1$, such that for any $n\in\mathbb N^*$, and for all $s\in \{ 0,\ldots,m-1 \},$
	\begin{equation}\label{eq:ergodoto}
	\sup_\alpha\sup_{\rhoalpha}\Big\|\frac1m\sum_{r=0}^{m-1} \Phi^{mn+s+r}(\rhoalpha)-\rhoalphainfty\Big\|_1\leq c\lambda^n.
 \end{equation}
 In particular there exist two positive constants $\tilde c$ and $\tilde{\lambda}<1$ such that for all $k\in\mathbb N$\begin{equation}\label{eq:ergodoto2}
	\sup_\alpha\sup_{\rhoalpha}\Big\|\frac1m\sum_{r=0}^{m-1} \Phi^{k+r}(\rhoalpha)-\rhoalphainfty \Big\|_1\leq \tilde c\tilde\lambda^{k}.
 \end{equation}
\end{corollary}
\begin{proof}
One can apply Theorem \ref{thm:ergod} on each $\mathcal H_\alpha$ since $\Phi_{\vert\mathcal H_\alpha}$ has a unique invariant state. One can then associate a period $m_\alpha$ to each $\mathcal H_\alpha$ as well as constants $c_\alpha$ and $\lambda_\alpha$. Then define 
$$m=lcm\{m_\alpha,\alpha=1,\ldots,\ell\},\quad c=\max_{\alpha}{c_\alpha},\quad \lambda=\max_{\alpha}{\lambda_\alpha},$$
where $lcm$ stands for lowest common multiple. Inequality \eqref{eq:ergodoto} is then satisfied with the above constants. For inequality \eqref{eq:ergodoto2}, for $k\in\mathbb N^*,$ take the Euclidiean division by $m,$ then there exist $q=\lfloor k/m\rfloor$ and $s\in\{ 0,\ldots,m-1 \}$ such that $k=qm+s$. Then
\begin{eqnarray*}
    \sup_\alpha\sup_{\rhoalpha}\Big\|\frac1m\sum_{r=0}^{m-1} \Phi^{k+r}(\rhoalpha)-\rhoalphainfty\Big\|_1
    &=&\sup_\alpha\sup_{\rhoalpha}\Big\|\frac1m\sum_{r=0}^{m-1} \Phi^{qm+s+r}(\rhoalpha)-\rhoalphainfty\Big\|_1\\
    &\leq& c\lambda^q\leq \tilde{c}\tilde\lambda^k,
\end{eqnarray*}
with $\tilde{c}=c/\lambda$ and $\tilde\lambda=\lambda^{1/m}.$
\end{proof}

Coming back to \eqref{eq:equality}, we have for all $k$
\begin{equation*}
    \tr{V_I\frac1m\sum_{r=0}^{m-1}\Phi^{mk+r}(\rho^{(\alpha)})V_I^\dag}=\tr{V_I\frac1m\sum_{r=0}^{m-1}\Phi^{mk+r}(\rho^{(\beta)})V_I^\dag}
\end{equation*}
which yields when $k\to\infty$
$$\tr{V_I \rho_{\infty}^{(\alpha)} V_I^\dag} = \tr{V_I \rho_{\infty}^{(\beta)} V_I^\dag},$$
contradicting \textbf{(ID)}. Then $\mathbb P_{\rhoalpha}\neq\mathbb P_{\rhobeta}$ for all $\rhoalpha$ and all $\rhobeta$. Therefore, assumption \textbf{(ID)}, which expresses the identifiability of each block with respect to the invariant states, implies actually a stronger identifiability condition. 

In the sequel, we shall reinforce this result by discussing the size of words that implies identifiability. Roughly speaking, when establishing exponential convergence, we shall need that the identifiability condition is obtained for a fixed size of words. The precise statement is expressed in the next proposition. 
\begin{proposition} \label{IDmod}
    There exists an integer N such that for all $\alpha\neq\beta$, for all $\rhoalpha$, $\rhobeta$  and for all $n\geqslant N$,  there exists a word $I_n\in\mathcal O^n$ such that
    $$\tr{V_{I_n} \rhoalpha V_{I_n}^\dag} \neq \tr{V_{I_n} \rhobeta V_{I_n}^\dag}.$$

\end{proposition}
\begin{proof}  
    Before we start, let us remark that if $I$ of size $|I|=n$ is such that $\operatorname{tr}(V_I\rho^{(\alpha)}V_I^\dag)\neq \operatorname{tr}(V_I\rho^{(\beta)}V_I^\dag)$, then for any $n'\geq n$, there exists $I'$ of size $|I'|=n'$ such that $\operatorname{tr}(V_{I'}\rho^{(\alpha)}V_{I'}^\dag)\neq \operatorname{tr}(V_{I'}\rho^{(\beta)}V_{I'}^\dag)$. Indeed, consistency implies,
    $$\sum_{J\in \mathcal O^{n'-n}}\tr{V_JV_I\rho^{(\alpha)}V_I^\dag V_J^\dag}=\tr{V_I\rho^{(\alpha)}V_I^\dag}$$
    and similarly with $\alpha$ replaced by $\beta$. If for all $J$, $\operatorname{tr}(V_JV_I\rho^{(\alpha)}V_I^\dag V_J^\dag)=\operatorname{tr}(V_JV_I\rho^{(\beta)}V_I^\dag V_J^\dag)$, then we have $\operatorname{tr}(V_I\rho^{(\alpha)}V_I^\dag )=\operatorname{tr}(V_I\rho^{(\beta)}V_I^\dag)$. Hence, there must exist $I'$, the concatenation of $I$ with a $J,$ such that $\operatorname{tr}(V_{I'}\rho^{(\alpha)}V_{I'}^\dag)\neq \operatorname{tr}(V_{I'}\rho^{(\beta)}V_{I'}^\dag)$.
    
    It is thus sufficient to prove that there exists an integer $N$ such that for any $\alpha\neq\beta$, $\rho^{(\alpha)}$ and $\rho^{(\beta)}$, there exists $I'\in \cup_{n=1}^N\mathcal O^n$ such that
    $$\tr{V_{I'}\rho^{(\alpha)}V_{I'}^\dag}\neq \tr{V_{I'}\rho^{(\beta)}V_{I'}^\dag}.$$
    
    We start with the case $m=1$, the proof for the case $m>1$ will be sketched at the end. Let $\alpha\neq\beta$ and let the word $I$ corresponding to assumption \textbf{(ID)} for the couple $(\alpha,\beta)$. We denote
    $$e_{\alpha,\beta}=\left\vert \tr{V_I \rhoalphainfty V_I^\dag} - \tr{V_I \rhobetainfty V_I^\dag}\right\vert>0.$$
    Using Theorem \ref{thm:ergod}, we have for all $k\in\mathbb N$
    \begin{eqnarray*}
    &&\sup_\alpha\sup_{\rhoalpha}\left\vert\tr{V_I\Phi^k(\rhoalpha)V_I^\dag}-\tr{V_I\rhoalphainfty V_I^\dag}\right\vert\\
    &=&\sup_\alpha\sup_{\rhoalpha}\left\vert\tr{ V_I\left(\Phi^k(\rhoalpha)-\rhoalphainfty \right) V_I^\dag }\right\vert\\
    &\leq&\sup_\alpha\sup_{\rhoalpha}\left\vert\tr{\left\|\Phi^k(\rhoalpha)-\rhoalphainfty\right\|_1 V_I^\dag V_I}\right\vert\\
    &\leq&\sup_\alpha\sup_{\rhoalpha}\Big\| \Phi^k(\rhoalpha)-\rhoalphainfty \Big\|_1\\&\leq&c\lambda^k.
    \end{eqnarray*}
    Here we have successively used that $\vert\tr{AB^\dag B}\vert\leq \tr{\vert A\vert B^\dag B}$ and next $V_I^\dag V_I\leq I_d$ (for operators, $A\leq B$ means $B-A\geq 0$). Therefore there exists $k\in\mathbb N^*$ such that
    \begin{eqnarray*}
       \sup_{\rhoalpha}\left\vert\tr{V_I\Phi^k(\rhoalpha)V_I^\dag}-\tr{V_I\rhoalphainfty  V_I^\dag}\right\vert &\leq& \frac{e_{\alpha,\beta}}{4}\\
       \sup_{\rhobeta}\left\vert\tr{V_I\Phi^k(\rhobeta)V_I^\dag}-\tr{V_I\rhobetainfty   V_I^\dag}\right\vert &\leq& \frac{e_{\alpha,\beta}}{4}.
    \end{eqnarray*}
    This way for all $\rhoalpha$ and all $\rhobeta$
    \begin{eqnarray*}
        &&\left\vert\tr{V_I\Phi^k(\rhoalpha)V_I^\dag}-\tr{V_I\Phi^k(\rhobeta)V_I^\dag}\right\vert\\
        &=&\left\vert\tr{V_I\Phi^k(\rhoalpha)V_I^\dag}-\tr{V_I\rhoalphainfty  V_I^\dag}+\tr{V_I\rhoalphainfty  V_I^\dag}-\tr{V_I\rhobetainfty   V_I^\dag}\right.\\&&\left.+\tr{V_I\rhobetainfty   V_I^\dag}-\tr{V_I\Phi^k(\rhobeta)V_I^\dag}\right\vert\\
        &\geq&-\left\vert\tr{V_I\Phi^k(\rhoalpha)V_I^\dag}-\tr{V_I\rhoalphainfty  V_I^\dag}\right\vert+\left\vert\tr{V_I\rhoalphainfty  V_I^\dag}-\tr{V_I\rhobetainfty   V_I^\dag}\right\vert\\&&-\left\vert\tr{V_I\rhobetainfty   V_I^\dag}-\tr{V_I\Phi^k(\rhobeta)V_I^\dag}\right\vert\\
        &\geq&\frac{e_{\alpha,\beta}}{2}.
    \end{eqnarray*}
    As a consequence, there exists a word $J$ (depending on $\rhoalpha$ and $\rhobeta$) of size $k$ such that
    $$\left\vert\tr{V_IV_J\rhoalpha V_J^\dag V_I^\dag}-\tr{V_IV_J\rhobeta V_J^\dag V_I^\dag}\right\vert>0.$$
    Indeed, if it was false, summing over $J$ would contradict the last inequality. Note that the important fact is that the size $k$ does not depend on $\rhoalpha$ and $\rhobeta$. 
    
    Then for the couple $(\alpha,\beta)$ there exists $N_{\alpha,\beta}=k+\vert I\vert$ such that for all $\rhoalpha$ and all $\rhobeta,$ there exists a word $I'$ of size $N_{\alpha,\beta}$ such that
    $$\left\vert\tr{V_{I'}\rhoalpha V_{I'}^\dag}-\tr{V_{I'}\rhobeta V_{I'}^\dag }\right\vert>0.$$
    Taking $N=\max_{\alpha\neq\beta}\{N_{\alpha,\beta}\}$ yields the result. 
    
    The proof for $m>1$ is similar replacing $\Phi^k$ by the average \eqref{eq:ergodoto2}. Then there exists $k$ such that 
    for all $\rhoalpha$ and $\rhobeta$
    $$\left\vert\tr{V_I\frac{1}{m}\sum_{r=0}^{m-1}\Phi^{k+r}(\rhoalpha)V_I^\dag}-\tr{V_I\frac{1}{m}\sum_{r=0}^{m-1}\Phi^{k+r}(\rhobeta)V_I^\dag}\right\vert>\frac{e_{\alpha,\beta}}{2}.$$ This implies that there exists a $J\in\cup_{r=0}^{m-1}\mathcal O^{k+r}$ such that
   $$\left\vert\tr{V_IV_{J}(\rhoalpha)V_{J}^\dag V_I^\dag}-\tr{V_IV_{J}(\rhobeta)V_{J}^\dag V_I^\dag}\right\vert>0.$$
   Then concatenating $I$ and $J$ yields the desired result with $N_{\alpha,\beta}=k+m-1+|I|$ and $N=\max_{\alpha\neq\beta}\{N_{\alpha,\beta}\}.$
\end{proof}

\section{Exponential selection of minimal  invariant subspaces}\label{sec:openloop}
This section is devoted to showing that the selection of minimal invariant subspace is exponentially fast in mean. This definitively improves the result of the literature where, as far as we know, only almost sure considerations have been developed \cite{MR4318593,BePelLin}. To this end, we use a Lyapunov function defined as
\begin{equation}
W(\rho)=\frac{1}{2}\sum_{\alpha\neq \beta}\sqrt{\tr{M_\alpha \rho} \tr{M_\beta \rho},}
\label{Lyap}
\end{equation}
for all $\rho\in \mathcal D$. Since $\sum_\alpha\tr{M_\alpha \rho}= 1$ and 
$ab \leq (a^2+b^2)/2$, then, for all $\rho\in\mathcal D,$
$$W(\rho)\leq \frac{\ell-1}{2}.$$
In the sequel, for $\alpha\neq\beta$, we define 
$$\mathcal A_{\alpha,\beta}=\{\rho\in\mathcal D \mid \tr{M_\alpha\rho}\neq0,~\tr{M_\beta\rho}\neq0\}.$$
 For any $\rho\in\mathcal A_{\alpha,\beta}$, denote $\tilde{\rho}^{(\alpha)}=\frac{M_\alpha\rho M_\alpha}{\tr{M_\alpha\rho}}$ and $\tilde{\rho}^{(\beta)}=\frac{M_\beta\rho M_\beta}{\tr{M_\beta\rho}}$.\\
 
\noindent The following lemma will be essential to prove the exponential convergence.

\begin{lemma} Let $N$ be the integer in Proposition \ref{IDmod}. We have
    $$\kappa=\sup_{\alpha\neq\beta}\sup_{\rho\in\mathcal A_{\alpha,\beta}}\sum_{I\in\mathcal O^N} \sqrt{ \tr{V_I \tilde{\rho}^{(\alpha)} V_I^\dag} \tr{V_I \tilde{\rho}^{(\beta)} V_I^\dag }}<1.$$

    \label{lem:rate}
\end{lemma}
\begin{proof}
Let $\rho\in\mathcal A_{\alpha,\beta}$. By Cauchy-Schwarz inequality, we get 
\begin{eqnarray*}
\sum_{I\in\mathcal O^N} \sqrt{ \tr{V_I \tilde{\rho}^{(\alpha)} V_I^\dag} \tr{V_I \tilde{\rho}^{(\beta)} V_I^\dag }} &\leq& \sqrt{ \sum_{I\in\mathcal O^N} \tr{V_I \tilde{\rho}^{(\alpha)} V_I^\dag}}  \sqrt{ \sum_{I\in\mathcal O^N} \tr{V_I \tilde{\rho}^{(\beta)} V_I^\dag}}\nonumber\\&=&\sqrt{ \tr{\Phi^N(\tilde{\rho}^{(\alpha)})}}  \sqrt{ \tr{\Phi^N(\tilde{\rho}^{(\beta)})}}\nonumber\\&=&1.
\end{eqnarray*}
By exploiting the equality case in Cauchy-Schwarz inequality, we have $$\sum_{I\in\mathcal O^N} \sqrt{ \tr{V_I \tilde{\rho}^{(\alpha)} V_I^\dag} \tr{V_I \tilde{\rho}^{(\beta)} V_I^\dag }}=1$$ if and only if there exists $\lambda \in \mathbb R$ such that
$$\forall I\in\mathcal O^N ,~~ \tr{V_I \tilde{\rho}^{(\alpha)} V_I^\dag} = \lambda \tr{ V_I \tilde{\rho}^{(\beta)} V_I^\dag}.$$
Summing over $I \in \mathcal O^N$, we see that necessarily $\lambda=1$, i.e.
$\forall I \in \mathcal O^N ,~~\operatorname{tr}(V_I \tilde{\rho}^{(\alpha)} V_I^\dag) = \operatorname{tr}( V_I \tilde{\rho}^{(\beta)} V_I^\dag).$
This is in contradiction with Proposition \ref{IDmod}. We can then conclude that for all $\rho\in\mathcal A_{\alpha,\beta}$
$$\sum_{I\in\mathcal O^N} \sqrt{ \tr{V_I \tilde{\rho}^{(\alpha)} V_I^\dag} \tr{V_I \tilde{\rho}^{(\beta)} V_I^\dag }}<1.$$
Since 
$(\tilde{\rho}^{(\alpha)},\tilde{\rho}^{(\beta)})\mapsto \sum_{I\in\mathcal O^N} \sqrt{ \operatorname{tr}(V_I \tilde{\rho}^{(\alpha)} V_I^\dag) \operatorname{tr}(V_I \tilde{\rho}^{(\beta)} V_I^\dag )}$
is a continuous function on a compact space, it reaches its supremum, hence, $\kappa<1$.
\end{proof}

Before expressing the main result of this section, let us present a lemma that will be used in the proof of the main result.
\begin{lemma}\label{lem:martin1}
    The process $(W(\rho_n))$ satisfies, for all $n\in\mathbb N,$
    $$\mathbb E[W(\rho_{n+1})\vert\rho_n]\leq W(\rhon).$$
\end{lemma}
\begin{proof}
    Let $k\in\mathbb N$, we have
   \begin{align*}
  \EE{W(\rho_{k+1})}{\rhok}&=\sum_{i\in\mathcal O} W \o \frac{V_i \rhok V_i^\dag}{\tr{V_i \rhok V_i^\dag}} \c \tr{V_i \rhok V_i^\dag}\\
  &=\sum_{i\in\mathcal O} W \o V_i \rhok V_i^\dag \c \\
  &= \frac{1}{2} \sum_{i\in\mathcal O} \sum_{\alpha \neq \beta} \sqrt{\tr{M_\alpha V_i \rhok V_i^\dag} \tr{M_\beta V_i \rhok V_i^\dag}} \\
  &=\frac{1}{2} \sum_{\alpha \neq \beta} \sum_{i\in\mathcal O}  \mathds 1_{\rhok\in\mathcal A_{\alpha,\beta}}\sqrt{\frac{\tr{ M_\alpha V_i \rhok V_i^\dag }}{\tr{M_\alpha \rhok}} \frac{\tr{M_\beta V_i \rhok V_i^\dag}}{\tr{M_\beta \rhok}}} \sqrt{\tr{M_\alpha \rhok} \tr{M_\beta \rhok} } \\
  &\leqslant  W(\rhok).
\end{align*}
The last inequality is obtained by using Cauchy-Schwarz inequality as in the previous lemma. 
\end{proof}


Exponential convergence could be derived directly if an inequality of the following type
\begin{equation}\label{eq:unpas}
   \mathbb  E[W(\rho_{n+1})]\leq \kappa\mathbb E[W(\rho_{n})],~ n\in\mathbb N
\end{equation}
would be satisfied with a constant $\kappa<1$. Indeed this would imply that for all $n\in\mathbb N$
$$\mathbb E[W(\rho_{n})]\leq \kappa^n\frac{(\ell-1)}{2}.$$
As we shall see, obtaining such inequality in our context would ask for an identification with one letter, i.e. for all $\alpha\neq\beta$ and for all $\rhoalpha$ and all $\rhobeta,$ there exists $i\in\mathcal O$ such that $\operatorname{tr}(V_i\rhoalpha V_i^\dag)\neq \operatorname{tr}(V_i\rhobeta V_i^\dag)$. This is a strong assumption that is in general not satisfied. Note that this is however the case in the QND models. Nevertheless, we can still obtain exponential convergence, expressed in the following theorem, which is the main result of this section.

\begin{theorem}\label{thm:expselect}Consider the quantum trajectory $(\rho_n)$ defined in $\eqref{traj}$ and the Lyapunov function $W(\rho)$ defined in $\eqref{Lyap}$. Suppose that $\textbf{(ID)}$ holds. Then there exist some constants $C>0$ and $\gamma>0$ such that for all $ n \geqslant0, $ 
    $$\E{W(\rho_n)} \leqslant C e^{-\gamma n  }.$$
\end{theorem}

\begin{proof}
To overcome inequality \eqref{eq:unpas}, which cannot be obtained for one letter, we consider the shift of $(\rho_n)$ with the length of identifiability of Proposition \ref{IDmod}. 

Let $N$ be the integer of Proposition \ref{IDmod}, we have for all $k\in\mathbb N$
\begin{align*}
  \EE{W(\rho_{k+N})}{\rhok}&=\sum_{I\in\mathcal O^N} W \o \frac{V_I \rhok V_I^\dag}{\tr{V_I \rhok V_I^\dag}} \c \tr{V_I \rhok V_I^\dag}\\
  &=\sum_{I\in\mathcal O^N} W \o V_I \rhok V_I^\dag \c \\
  &= \frac{1}{2} \sum_{I\in\mathcal O^N} \sum_{\alpha \neq \beta} \sqrt{\tr{M_\alpha V_I \rhok V_I^\dag} \tr{M_\beta V_I \rhok V_I^\dag}} \\
  &=\frac{1}{2} \sum_{\alpha \neq \beta} \sum_{I\in\mathcal O^N}  \mathds 1_{\rhok\in\mathcal A_{\alpha,\beta}}\sqrt{\frac{\tr{M_\alpha V_I \rhok V_I^\dag}}{\tr{M_\alpha \rhok}} \frac{\tr{M_\beta V_I \rhok V_I^\dag}}{\tr{M_\beta \rhok}}} \sqrt{\tr{M_\alpha \rhok} \tr{M_\beta \rhok} } \\
  &\leqslant \kappa W(\rhok),
\end{align*}
where $\kappa<1$ is the constant appearing in Lemma \ref{lem:rate}. Taking expectation, we get for all $k\in\mathbb N$
$$\mathbb E[W(\rho_{k+N})]\leq \kappa\mathbb E[W(\rho_{k})].$$
Then for all $q\in\mathbb N$ and all $k\in\mathbb N$
$$\mathbb E[W(\rho_{k+qN})]\leq \kappa^q\mathbb E[W(\rho_{k})].$$
Now let $n\geq N.$ There exist $q,p$ with $0\leq p<N$ such that $n=qN+p$. This implies 

\begin{eqnarray*}
E[W(\rho_{n})]&\leq& \kappa^{\lfloor\frac{n}{N}\rfloor}\max_{k=0,\ldots,N-1}\{\mathbb E[W(\rho_k)]\}\\
&\leq& \kappa^{\frac nN-1}\frac{(\ell-1)}{2}.
\end{eqnarray*}
Then we have shown the exponential convergence after time $N$ with $\gamma=-\ln(\kappa^{1/N})$ and a constant $\kappa^{-1}(\ell-1)/2.$ Up to changing the constant in front, we have the exponential convergence for all time by gathering all the steps before time $N$.
\end{proof}

\begin{corollary}
   The selection of minimal invariant subspace occurs exponentially almost surely. More precisely, let $\gamma$ be the parameter of Theorem \ref{thm:expselect}. Then for all $0<\lambda<\gamma$, we have almost surely 
   $$\lim_{n\to\infty}e^{\lambda n}W(\rhon)=0.$$
\end{corollary}

\begin{proof}
This is an easy application of Borel-Cantelli lemma. For all $\varepsilon>0$, using the Markov inequality, we have that
$$\mathbb P[e^{\lambda n}W(\rhon)>\varepsilon]\leq \frac{\mathbb E[e^{\lambda n}W(\rhon)]}{\varepsilon}\leq\frac{Ce^{-(\gamma-\lambda) n}}{\varepsilon}.$$
This way
$$\sum_{n\geq 1}\mathbb P[e^{\lambda n}W(\rhon)>\varepsilon]<\infty$$
and by Borel-Cantelli lemma
$$\mathbb P\left[\bigcup_{N\geq0}\bigcap_{n\geq N}e^{\lambda n}W(\rhon)\leq\varepsilon\right]=1,$$
which implies that almost surely
$$\lim_{n\to\infty}e^{\lambda n}W(\rhon)=0.$$
\end{proof}


Having observed that exponential convergence is a consequence of the identifiability condition,  
the next question that arises is whether we can force quantum trajectories to reach a target subspace by applying a control input. This is the topic of the next section.

\section{Feedback stabilization of subspaces}\label{sec:feedback}
In this section, the goal is to stabilize one of the minimal invariant subspaces. We denote the index of this subspace by $\bar\alpha.$
In general, the controlled dynamics consists of two steps from $\rho_n$ to $\rho_{n+1}$: first a measurement, then a feedback control action. For the sake of simplicity, we denote the intermediate measurement step by
\begin{equation}\label{inter_step}
    {\rho}_{n+\frac{1}{2}}=\frac{V_{i_n}{\rho}_{n}V_{i_n}^\dag}{\operatorname{tr}\left(V_{i_n}{\rho}_{n}V_{i_n}^\dag\right)}.
\end{equation}
It is followed by the feedback control step given by
\begin{equation}\label{control_step}
    {\rho}_{n+1}= U(u_n) \rho_{n+\frac{1}{2}} U(u_n)^\dag,
\end{equation}
where $U(u)=e^{-iuH}$ and $H$ denotes the so-called control Hamiltonian operator which is Hermitian. The quantity $u_n$ is a real scalar corresponding to the control intensity applied at step $n$.

Let us define, for all $\rho\in\mathcal D,$
$$V(\rho)=\sqrt{1 - \tr{\Malphabar \rho}} \text{\quad and \quad } R(\rho)=\sum_{\beta\neq \alphabar}\sqrt{\tr{M_\beta\rho}}.$$
We consider the following Lyapunov function
$$ Z(\rho)= V(\rho) + \varepsilon R(\rho),$$
where $\varepsilon$ is a small parameter to be specified later. The feedback control we will consider is the following:
\begin{equation}
    u_{n}= \begin{cases}\underset{u \in [-\ubar, \ubar]}{\argmin} \EE{Z(U(u) \rho_{n+\frac{1}{2}} U(u)^\dag   )}{ \rho_{n-N+1}} & \text { if } n= qN -1 ,~q\in \mathbb N   \\ 0 & \text { if } n\neq qN-1,~q\in \mathbb N  \end{cases}
    \label{feedback}
\end{equation}
  
In essence, the control strategy consists in acting on the system every $N$ steps, where $N$ is the length of identifiability of Proposition \ref{IDmod}.
\begin{remark}
The integer $N$ used in the aforementioned control strategy can be substituted with any value $N'$ such that $N' \geq N$, since the identifiability result remains valid, as demonstrated in the proof of Proposition \ref{IDmod}. \end{remark}
Now, let us introduce an assumption concerning the control Hamiltonian, which can be interpreted as a controllability assumption.
\begin{assumption}
    Let  $\{ \ket{\phi_j}, j=1,\ldots,d_{\bar\alpha} \}$ be an orthonormal basis of $\mathcal H_{\bar\alpha}$, $$\textrm{Vect}\{H^k\ket{\phi_j},~k=1,\ldots,d,~j=1,\ldots,d_{\bar\alpha}\}=\bigoplus_{\beta \neq \alphabar} \H_\beta.$$
    \label{controllability}
\end{assumption}
The assumption does not depend on a particular choice of an orthonormal basis. If it holds for one basis, it holds for any basis. We stop the power of matrix $H^k$ for $k=d$ since by Cayley-Hamilton theorem, any upper power $H^k,\,k>d,$ can be expressed as a polynomial of $H$ of degree less than $d.$  

Our main result is the following theorem.

\begin{theorem}\label{thm:control}
 Suppose that Assumption \eqref{controllability} holds. The feedback scheme \eqref{feedback} exponentially stabilizes in mean the subspace $\Halphabar$ in the following sense: there exist $\bar C>0$ and $\bar \gamma>0$ such that for all $n \geqslant 0$
 $$\E{\sqrt{1-\tr{M_{\alphabar} \rho_n}}} \leqslant \bar C e^{-\bar \gamma n  }.$$
\end{theorem}

As in the case without control, an almost sure convergence can also be concluded, whose proof is strictly the same as the one without feedback.

\begin{corollary}
    The feedback scheme stabilizes the subspace $\mathcal H_{\bar \alpha}$ exponentially almost surely. More precisely, for all $0<\lambda<\bar\gamma,$ we have almost surely
    $$\lim_{n\to\infty}e^{\lambda n}\sqrt{1-\tr{M_{\bar \alpha}\rhon}}=0.$$ 
\end{corollary}
Before presenting the proof of the theorem, we shall need the following technical lemmas.

\subsection{Technical lemmas}
We set up the notation
$$S_{<\delta}=\{\rho\in\mathcal D\mid \tr{\Malphabar \rho}<\delta\},$$
for $\delta>0.$ Similar notation shall be used with $>$, $\leq$ and $\geq.$

\begin{lemma}\label{lem:control}
There exists $\delta_0>0$ such that
$$
\inf _{\rho \in S_{<\delta_0}} \max_{u \in [-\ubar, \ubar]} \operatorname{Tr}\left(\Malphabar U(u) \Phi^N ( \rho ) U(u)^\dag \right) > 2 \delta_0.
$$
\end{lemma}
\begin{proof}
    The proof is very similar to \cite[Lemma 1 $\&$ Corollary 1]{mirrahimi2009feedback}, with a slight adaptation that we make precise here. Define $ S_0=\{ \rho\in\mathcal D \mid \tr{\Malphabar\rho}=0\},$ i.e. the set of states supported by $\bigoplus_{\beta \neq \alphabar} \H_\beta$. Let us first show  that there exists $\delta_0>0$ such that
    $$\inf _{\rho \in S_0} \max_{u \in [-\ubar, \ubar]} \operatorname{Tr}\left(\Malphabar U(u) \Phi^N ( \rho ) U(u)^\dag \right) > 2 \delta_0.$$
    By contradiction assume that $$\inf _{\rho \in S_0} \max_{u \in [-\ubar, \ubar]} \operatorname{Tr}\left(\Malphabar U(u) \Phi^N ( \rho ) U(u)^\dag \right)=0.$$ Using the continuity of the application $ \rho \mapsto \max_{u \in [-\ubar, \ubar]} \tr{\Malphabar U(u)  \Phi^N(\rho) U(u) ^\dag} $ and the compactness of $S_0,$ this implies that there exists $\rho\in S_0$ such that for all $u \in [-\ubar,\ubar],$ 
    \begin{equation}\tr{\Malphabar U(u) \Phi^N(\rho) U(u)^\dag } = 0.
    \label{eq:condition}
    \end{equation}
    Since $\rho$ is supported by $\bigoplus_{\beta \neq \alphabar} \H_\beta$, the state $\Phi^N(\rho)$ is still supported by $\bigoplus_{\beta \neq \alphabar} \H_\beta$. We can then write $$\Phi^N(\rho)= \sum_{i=1}^{d-d_{\alphabar}}\lambda_i \ketbra{\psi_i}$$ where $\{ \ket{\psi_i}, i=1,\ldots,d-d_{\alphabar} \}$ is an orthonormal basis of $\bigoplus_{\beta \neq \alphabar} \H_\beta $. We also write $$\Malphabar= \sum_{j=1}^{d_{\alphabar}} \ketbra{\phi_j}$$ where  $\{ \ket{\phi_j}, j=1,\ldots,d_{\bar\alpha} \}$ is an orthonormal basis of $\Halphabar.$ The condition \eqref{eq:condition}
     implies that
    $$\sum_{i=1}^{d-d_{\alphabar}}\lambda_i\tr{\Malphabar U(u) \ketbra{\psi_i} U(u)^\dag }=0.$$
    Now let $i$ be such that $\lambda_i \neq 0$. Since all the terms of the previous equation are non-negative, we have $$\tr{\Malphabar U(u) \ketbra{\psi_i} U(u)^\dag }=\sum_{j=1}^{d_{\alphabar}}\tr{\ketbra{\phi_j} U(u) \ketbra{\psi_i} U(u)^\dag }=0.$$ Again, all terms are non-negative, which implies that for all $j=1,\ldots,d_{\alphabar}$
    $$ \bra{\psi_i} U(u) \ket{\phi_j} = 0.$$
    Then differentiating $k$ times with respect to $u$ and evaluating this derivative at $u=0$, we get
    $$ \bra{\psi_i} H^k \ket{\phi_j} = 0$$
    for all $k \in \{1,\ldots,d \}$ and all $ j \in \{1,\ldots,d_{\alphabar} \}$. This leads to a contradiction since the family $H^k \ket{\phi_j}$ generates the subspace $\bigoplus_{\beta \neq \alphabar} \H_\beta.$ Therefore, for all $\rho \in S_0$,
    $$ \max_{u \in [-\ubar, \ubar]}\tr{\Malphabar U(u) \Phi^N(\rho) U(u)^\dag} > 0.$$
    By compactness of $S_0$, there exists $\upsilon >0$ such that 
    $$ \inf_{ \rho \in S_0    }\max_{u \in [-\ubar, \ubar]}\tr{\Malphabar U(u) \Phi^N(\rho) U(u)^\dag} \geq \upsilon >0.$$
    Moreover, by continuity of the functions $ \rho \mapsto \max_{u \in [-\ubar, \ubar]} \tr{\Malphabar U(u)  \Phi^N(\rho) U(u) ^\dag} $ and $ \rho \mapsto \tr{\Malphabar \rho}$ and compactness of the set of density matrices, there exists $\eta>0$ such that
    $$   \inf_{ \rho \in S_{<\eta}   }\max_{u \in [-\ubar, \ubar]}\tr{\Malphabar U(u) \Phi^N(\rho) U(u)^\dag} > \frac{\upsilon}{2}.$$
    Taking $\delta_0=min( \eta , \frac{\upsilon}{4})$ leads to the desired lemma.
\end{proof}

The next lemma is again a consequence of the identifiability assumption, and will be essential in the proof of the final theorem.
\begin{lemma}\label{lem:CS2}
    For all $n\in\mathbb N^*$, for all $\beta \in \{1,\ldots,\ell \}$ and for all $\rho$ such that $\tr{M_\beta\rho}\neq0$, we have
    $$\sum_{I\in\mathcal O^n}\sqrt{\tr{M_\beta V_I\tilde{\rho}^{(\beta)} V_I^\dag}\tr{V_I\rho V_I^\dag}}\leq 1.$$
   Furthermore
    $$\sum_{I\in\mathcal O^n}\sqrt{\tr{M_\beta V_I\tilde{\rho}^{(\beta)} V_I^\dag}\tr{V_I\rho V_I^\dag}}= 1$$
    if and only if 
    $$\tr{V_I\tilde{\rho}^{(\beta)}V_I^\dag}=\tr{V_I\rho V_I^\dag}$$
    for all $I\in\mathcal O^n.$
\end{lemma}
\begin{proof}
 Let $n\in\mathbb N^*$, let $\beta \in \{1,\ldots,\ell\}$ 
 and let $\rho$ be such that $\tr{M_\beta\rho}\neq0$.
 Once again, we shall use Cauchy-Schwarz inequality, which gives
    \begin{eqnarray*}
        \sum_{I\in\mathcal O^n}\sqrt{\tr{M_\beta V_I\tilde{\rho}^{(\beta)} V_I^\dag}\tr{V_I\rho V_I^\dag}}
        &\leq&\sqrt{\sum_{I\in\mathcal O^n} \tr{M_\beta V_I\tilde{\rho}^{(\beta)} V_I^\dag}}\sqrt{\sum_{I\in\mathcal O^n}\tr{V_I\rho V_I^\dag}}\\
        &=&\sqrt{\tr{M_\beta\Phi^n(\tilde{\rho}^{(\beta)})}}\sqrt{\tr{\Phi^n(\rho)}}\\
        &=&\sqrt{\tr{{\Phi^\dag}^n(M_\beta)\tilde{\rho}^{(\beta)})}}\\&=&\sqrt{\tr{M_\beta\tilde{\rho}^{(\beta)}}}=1.
    \end{eqnarray*}
    Then we have equality if there exists $\lambda$ such that for all $I\in\mathcal O^n$ we have $\operatorname{tr}(V_I\tilde{\rho}^{(\beta)}V_I^\dag)=\lambda\operatorname{tr}(V_I\rho V_I^\dag)$. Summing over $I$ yields $\lambda=1$, and the lemma is proved.
    \end{proof}

Now we are in the position to prove Theorem \ref{thm:control}.
\subsection{Proof of Theorem \ref{thm:control}}
\begin{customproof}[Proof]
    Let $q \in \mathbb N$. We will first show that
    $$ \EE{Z(\rho_{(q+1)N}}{\rho_{qN}} \leq \bar \kappa Z(\rho_{qN}), $$ for some $\bar \kappa <1.$ To this end, we consider different possibilities for $\rho_{qN}$ and split the study of the increment to whether it belongs to the sets 
    $S_{<\delta},~S_{> 1-\delta},
    $ or $S_{\geq\delta}\cap S_{\leq 1-\delta}$. Let us treat the three different cases separately.\\

\paragraph{The case $\rho_{qN} \in S_{<\delta}$} We have $R(\rho) \leq \sqrt{l-1},$ for all $\rho\in\mathcal D.$
Hence
\begin{align*}
    \EE{Z(\rho_{(q+1)N})}{\rho_{qN}}&=\min_{u \in [-\ubar, \ubar]}\EE{Z ( U(u) \rho_{(q+1)N-\frac{1}{2}}  U(u)^\dag    )}{\rho_{qN}}\\
    &\leq \min_{u \in [-\ubar, \ubar]}\EE{V ( U(u) \rho_{(q+1)N-\frac{1}{2}}  U(u)^\dag    )}{\rho_{qN}} + \varepsilon \sqrt{l-1}\\
    &=  \min_{u \in [-\ubar, \ubar]} \sum_{I\in\mathcal O^N} \tr{V_I \rho_{qN} V_I^\dag} \sqrt{1 -\frac{\tr{\Malphabar U(u) V_I \rho_{qN} V_I^\dag U(u)^\dag}}{\tr{V_I \rho_{qN} V_I^\dag}}}+ \varepsilon \sqrt{l-1}\\
    &\leq  \min_{u \in [-\ubar, \ubar]} \sqrt{1-\tr{\Malphabar U(u) \Phi^N(\rho_{qN}) U(u)^\dag}}+ \varepsilon \sqrt{l-1}
\end{align*}
by Cauchy-Schwarz inequality. Then 
\begin{align*}
    \EE{Z(\rho_{(q+1)N})}{\rho_{qN}}&\leq \sup_{\rho \in S_{<\delta}} \min_{u \in [-\ubar, \ubar]} \sqrt{1-\tr{\Malphabar U(u) \Phi^N(\rho ) U(u)^\dag}}+ \varepsilon \sqrt{l-1}\\
    &= \sqrt{1 - \inf _{\rho \in S_{<\delta}} \max_{u \in [-\ubar, \ubar]} \tr{\Malphabar U(u) \Phi^N(\rho ) U(u)^\dag}  }+ \varepsilon \sqrt{l-1}.
\end{align*}
Using Lemma \ref{lem:control}, as long as $\delta \leq \delta_0,$ we have\begin{align*}
    \EE{Z(\rho_{(q+1)N})}{\rho_{qN}} &\leq \sqrt{1-2\delta_0} + \varepsilon \sqrt{l-1}\\
    &\leq  \o \frac{\sqrt{1-2\delta_0}+\varepsilon\sqrt{l-1}}{\sqrt{1-\delta}} \c V(\rho_{qN})\\
    &\leq \o \frac{\sqrt{1-2\delta_0}+\varepsilon\sqrt{l-1}}{\sqrt{1-\delta}} \c  \o V(\rho_{qN})+\epsilon R(\rho_{qN}) \c\\
    &=\kappa_0 Z(\rho_{qN})
\end{align*}
with $\kappa_0 :=\frac{\sqrt{1-2\delta_0}+\varepsilon\sqrt{l-1}}{\sqrt{1-\delta}}$. In the second inequality, we used the fact that $\rho_{qN}$ belongs to $S_{<\delta}$, so $\frac{V(\rho_{qN})}{\sqrt{1-\delta}} \geq 1.$ We need to ensure that $\kappa_0 < 1.$ For this it is sufficient that 
$$ \varepsilon < \frac{\sqrt{1-\delta}-\sqrt{1-2\delta_0}}{\sqrt{l-1}},$$
and from now we set $\varepsilon>0$ as verifying this inequality.
\bigskip

\paragraph{The case $\rho_{qN} \in S_{\geq\delta}\cap S_{\leq 1-\delta}$} Let us first treat the evolution of $R(\rho)$ before the control step. For this function, using Lemma \ref{lem:CS2}, we obtain 
\begin{eqnarray*}
\EE{R(\rho_{(q+1)N-\frac{1}{2}})}{\rho_{qN}} &=&\sum_{\beta\neq\bar\alpha}\sum_{I\in\mathcal O^N}\sqrt{\tr{M_\beta\frac{V_I\rho_{qN} V_I^\dag}{\tr{V_I\rho_{qN} V_I^\dag}}}} \tr{V_I\rho_{qN} V_I^\dag}\\
&=&\sum_{\beta\neq\bar\alpha}\sum_{I\in\mathcal O^N}\sqrt{\tr{M_\beta V_I\rho_{qN} V_I^\dag} \tr{V_I\rho_{qN} V_I^\dag}}\nonumber\\
&=&\sum_{\beta\neq\bar\alpha}\mathds 1_{\tr{M_\beta\rho_{qN}}>0}\sum_{I}\sqrt{\tr{M_\beta V_I\tilde\rho^{(\beta)}_{kN} V_I^\dag} \tr{V_I\rho_{qN} V_I^\dag}}\sqrt{\tr{M_\beta\rho_{qN}}}\nonumber\\
&\leq&\sum_{\beta\neq\bar\alpha}\sqrt{\tr{M_\beta\rho_{qN}}}.\nonumber
\label{eq:sp}
\end{eqnarray*}
In the last line, equality is achieved if and only if  for all $\beta\neq\bar\alpha$ such that $\tr{M_\beta\rho_{qN}}>0,$ we can apply the equality case of Lemma \ref{lem:CS2}. Therefore for all $\beta\neq\bar\alpha$ such that $\tr{M_\beta\rho_{qN}}>0,$ we have that for all $I\in\mathcal O^N$  
\begin{equation}\label{eq:equality22}\tr{V_I\tilde\rho_{qN}^{(\beta)}V_I^\dag}=\tr{V_I\rho_{qN} V_I^\dag},
\end{equation}
where we recall that the notation $\tilde\rho^{(\beta)}$ holds for $\tilde\rho^{(\beta)}=\frac{M_\beta\rho M_\beta}{\tr{M_\beta\rho}},$ when $\tr{M_\beta\rho}\neq0.$ Let us check that it is not possible. We shall consider two situations. First suppose that there exist $\beta_1$ and $\beta_2$ with $\tr{M_{\beta_1}\rho_{qN}}>0$ and $\tr{M_{\beta_2}\rho_{qN}}>0.$ Then equality \eqref{eq:equality22} implies that for all $I\in\mathcal O^N$  
\begin{equation*}\label{eq:equality222}\tr{V_I\tilde\rho_{qN}^{(\beta_1)}V_I^\dag}=\tr{V_I\tilde\rho_{qN}^{(\beta_2)}V_I^\dag}.
\end{equation*} This is in contradiction with Proposition \ref{IDmod}. Then assume that there exists only one $\beta_1$ which verifies $\operatorname{tr}(M_{\beta_1}\rho_{qN})>0$ (the existence of $\beta_1$ is ensured by the fact that $\rho_{qN} \in S_{\geq\delta}\cap S_{\leq 1-\delta}$). Here we have $\operatorname{tr}(M_{\bar\alpha}\rho_{qN})+\operatorname{tr}(M_{\beta_1}\rho_{qN})=1$ which yields
$$\tr{V_I\tilde\rho_{qN}^{(\beta_1)}V_I^\dag}=\tr{V_I\rho_{qN} V_I^\dag}=\tr{M_{\bar\alpha}\rho_{qN}}\tr{V_I\tilde\rho_{qN}^{(\bar\alpha)} V_I^\dag}+\tr{M_{\beta_1}\rho_{qN}}\tr{V_I\tilde\rho_{qN}^{(\beta_1)} V_I^\dag}.$$ At this stage, we then have $$\tr{V_I\tilde\rho_{qN}^{(\bar\alpha)} V_I^\dag}=\tr{V_I\tilde\rho_{qN}^{(\beta_1)} V_I^\dag},$$ which once again contradicts Proposition \ref{IDmod}. This so proves a strict inequality in \eqref{eq:sp}. Now define on $S_{\geq\delta}\cap S_{\leq 1-\delta}$ the quantity $$\kappa(\rho)=\frac{\displaystyle{\sum_{\beta\neq\bar\alpha}\sum_{I\in\mathcal O^N}\sqrt{\tr{M_\beta V_I\rho V_I^\dag} \tr{V_I\rho V_I^\dag}}}}{R(\rho)},$$ which is well defined for $\rho\in S_{\geq\delta}\cap S_{\leq 1-\delta}$. The application $\rho\rightarrow \kappa(\rho)$ is continuous on $S_{\geq\delta}\cap S_{\leq 1-\delta}$ and one can define   $$\kappa_1:=\sup_{\rho\in S_{\geq\delta}\cap S_{\leq 1-\delta}}\kappa(\rho).$$ By compactness of the set  $S_{\geq\delta}\cap S_{\leq 1-\delta}$ and the continuity, we have $\kappa_1<1$. As a result,
$$\EE{R(\rho_{(q+1)N-\frac{1}{2}})}{\rho_{qN}} \leq \kappa_1 R(\rho_{qN}).$$
Now concerning the function $V,$ we have
$$\EE{V(\rho_{(q+1)N-\frac{1}{2}})}{\rho_{qN}} \leq V(\rho_{qN}).$$
Hence, for the sum $Z=V + \varepsilon R,$
$$\EE{Z(\rho_{(q+1)N-\frac{1}{2}})}{\rho_{qN}} \leq V(\rho_{qN}) + \kappa_1  \varepsilon R(\rho_{qN}).$$
The feedback control is minimizing the conditional expectation $\EE{Z(U(u)\rho_{(q+1)N-\frac{1}{2}} U(u)^\dag}{\rho_{qN}}$ over $u\in [-\ubar, \ubar]$. In particular, this conditional expectation is equal to $\EE{Z(\rho_{(q+1)N-\frac{1}{2}})}{\rho_{qN}}$ for $u=0.$ Consequently
\begin{align*}  
\EE{Z(\rho_{(q+1)N})}{\rho_{qN}} &\leq \EE{Z(\rho_{(q+1)N-\frac{1}{2}})}{\rho_{qN}} \\
&\leq V(\rho_{qN}) + \kappa_1  \varepsilon R(\rho_{qN}).
\end{align*}
Now we shall ensure that $ V(\rho_{qN}) + \kappa_1 \varepsilon R(\rho_{qN}) \leq \tilde \kappa_1 (   V(\rho_{qN}) + \varepsilon R(\rho_{qN}) )$ for some $\tilde \kappa_1 < 1$. The following framing relationship holds:
$$\frac{1}{\sqrt{\ell-1}} R(\rho) \leq  V(\rho) \leq R(\rho).$$
Then the expected inequality $ V(\rho_{qN}) + \kappa_1 \varepsilon R(\rho_{qN}) \leq \tilde \kappa_1 (   V(\rho_{qN}) + \varepsilon R(\rho_{qN}) )$ is obtained if
$$ V(\rho) \leq \frac{(\tilde \kappa_1 - \kappa_1) \varepsilon}{1-\tilde \kappa_1} R(\rho)$$
which will be verified if $\frac{(\tilde \kappa_1  -\kappa_1) \varepsilon }{1-\tilde \kappa_1} \geq 1.$ The smallest value for $\tilde \kappa_1$ satisfying this inequality is 
$$\tilde \kappa_1 = \frac{1+\varepsilon \kappa_1}{1 + \varepsilon} <1$$
giving us the desired relation
$$\EE{Z(\rho_{(q+1)N})}{\rho_{qN}} < \tilde \kappa_1 Z(\rho_{qN}).$$
\bigskip

\paragraph{The case $\rho_{qN} \in S_{> 1-\delta}$} Analysing again first the behaviour of $R(\rho)$ before the control step, we get
\begin{align*}
&\EE{R(\rho_{(q+1)N-\frac{1}{2}})}{\rho_{qN}}\\=&\sum_{I\in\mathcal O^N}\sum_{\beta\neq \bar\alpha}\sqrt{\tr{M_\beta V_I \rho_{qN} V_I^\dag}}\sqrt{\tr{V_I\rho_{qN} V_I^\dag}}\\
=&\sum_{I\in\mathcal O^N}\sum_{\beta\neq \bar\alpha}\sqrt{\tr{M_\beta V_I \rho_{qN} V_I^\dag}}\sqrt{\tr{M_{\bar\alpha}V_I\rho_{qN} V_I^\dag}+\tr{V_I\rho_{qN} V_I^\dag}-\tr{M_{\bar\alpha}V_I\rho_{qN} V_I^\dag}}\\
\leq& \sum_{I\in\mathcal O^N}\sum_{\beta\neq \bar\alpha}\sqrt{\tr{M_\beta V_I \rho_{qN} V_I^\dag}}\sqrt{\tr{M_{\bar\alpha}V_I\rho_{qN} V_I^\dag}}\\
\quad\quad&+\sum_{I\in\mathcal O^N}\sum_{\beta\neq \bar\alpha}\sqrt{\tr{M_\beta V_I \rho_{qN} V_I^\dag}}\sqrt{\tr{V_I\rho_{qN} V_I^\dag}-\tr{M_{\bar\alpha}V_I\rho_{qN} V_I^\dag}},\\
\end{align*}
where for the last inequality, we have used $\sqrt{a+b}\leq \sqrt{a}+\sqrt{b}.$ 
To bound the first term in the expression above, we use the same idea as previously applied in the proof of Theorem \ref{thm:expselect}. Similarly to the definition of $\kappa$ in Lemma \ref{lem:rate}, we set $$\kappa'=\sup_{\beta\neq\bar\alpha}\sup_{\rho\in\mathcal A_{\bar\alpha,\beta}}\sum_{I\in\mathcal O^N} \sqrt{ \tr{V_I \rho^{(\bar\alpha)} V_I^\dag} \tr{V_I \rho^{(\beta)} V_I^\dag }}\leq \kappa<1.$$ The second term can be bounded by applying Cauchy-Schwarz inequality. Then we get
\begin{align*}
&\EE{R(\rho_{(q+1)N-\frac{1}{2}})}{\rho_{qN}}\\
\leq\,& \kappa' R(\rho_{qN})+\sum_{\beta\neq \bar\alpha}\sqrt{\sum_{I\in\mathcal O^N}\tr{M_\beta V_I \rho_{qN} V_I^\dag}}\sqrt{\sum_{I\in\mathcal O^N}\left(\tr{V_I\rho_{qN} V_I^\dag}-\tr{M_{\bar\alpha}V_I\rho_{qN} V_I^\dag}\right)} \\
=\,&\kappa'R(\rho_{qN})+\sum_{\beta\neq\bar\alpha}\sqrt{\tr{M_\beta \rho_{qN}}}\sqrt{1-\tr{M_{\bar\alpha}\rho_{qN}}}\\
\leq\,& \kappa' R(\rho_{qN})+\sqrt\delta R(\rho_{qN})\\
=\,&\kappa_2R(\rho_{qN})
\end{align*}
where  $\kappa_2:=\kappa'+\sqrt\delta.$ Here appears a new constraint on the parameter $\delta$ in order to have $\kappa_2 < 1:$ it must satisfy $\sqrt{\delta} < 1 - \kappa'.$ Finally, combining this with the constraint deduced in the first case $\rho_{qN} \in S_{<\delta}$, we set $\delta>0$ as verifying $\sqrt\delta<\min(\sqrt{\delta_0},1-\kappa').$

Hence again we have 
$$\EE{Z(\rho_{(q+1)N-\frac{1}{2}})}{\rho_{qN}} \leq V(\rho_{qN}) + \kappa_2  \varepsilon R(\rho_{qN})$$
and we can repeat the calculations made in the previous paragraph to state that
$$\EE{Z(\rho_{(q+1)N})}{\rho_{qN}} \leq \tilde \kappa_2 Z(\rho_{qN})$$
with $\tilde \kappa_2 = \frac{1+\varepsilon \kappa_2 }{1 + \varepsilon} <1. $
\bigskip
 
\paragraph{Conclusion} Combining the outcomes obtained for the three different cases, we can conclude that for all $q \in \mathbb N,$
$$\EE{Z(\rho_{(q+1)N})}{\rho_{qN}} \leq \bar \kappa Z(\rho_{qN})$$
with $\bar \kappa = max(\kappa_0, \tilde \kappa_1, \tilde \kappa_2)<1.$
Then there exists $  C >0$ such that $\E{Z(\rho_{qN})} \leq   C {\bar \kappa}^q$ for all $q \in \mathbb N$. Now for a generic $n \in \mathbb N$ whose Euclidean division by $N$ is $n=qN+p,$ with $q=\lfloor \frac{n}{N} \rfloor$ and $0 \leq p < N,$ we have
\begin{align*}
  \EE{\tr{M_\alpha \rho_{n}}}{\rho_{qN}}&=\sum_{I\in\mathcal O^p} \tr{ M_\alpha \frac{ V_I \rho_{qN} V_I^\dag}{\tr{V_I \rho_{qN} V_I^\dag}} } \tr{V_I \rho_{qN} V_I^\dag}\\
  &= \tr{\Phi^\dag( M_\alpha  ) \rho_{qN}}\\
  &= \tr{M_\alpha \rho_{qN}}.
\end{align*}
    As $Z$ is a concave function, by Jensen's inequality we get
$$\mathbb E[Z(\rho_n)\vert \rho_{qN}]\leq Z(\rho_{qN}).$$
Hence
\begin{align*}
\E{Z( \rho_{n})} &\leq \E{Z( \rho_{qN})}\\
&\leq  C {\bar \kappa}^{q}\\
&\leq  C {\bar \kappa}^{ \frac{n}{N} - 1 }.
\end{align*}
Consequently,
$$ \E{\sqrt{1-\tr{\Malphabar \rho_n}}}  \leq  C {\bar \kappa}^{ \frac{n}{N} - 1 }. $$
Thus, we have shown the exponential convergence with $\bar \gamma = -\ln(\bar \kappa^{\frac{1}{N}} )$ and $\bar C =  C \kappa^{-1}.$
\end{customproof}

\section{Conclusion}
The large-time behavior of quantum trajectories has been the subject of numerous developments starting from the pioneering work \cite{KuMA}. From a Markov point of view, invariant measure and convergence have been studied in \cite{Barchinv,benoist2019invariant,Benoistinvcont}. This has led to results on limit theorems for quantum trajectories \cite{Bjanluka}. A large number of contributions has also been derived for the processes of measurement outcomes, limit theorems \cite{Carbonecentral,sabot,BePelLin,van2015sanov,Carbone2014HomogeneousOQ}, ergodic theory \cite{KuMa1,Benoistentrop1,Benoistentrop2}, concentration of measure \cite{Girconcentr,Benoist2021DeviationBA}. All these works are essentially based on a Markovian approach to the large-time behavior problem. 

The current work contributes to the understanding of the large-time behavior of quantum trajectories with an approach based on Lyapunov functions. Our work has focused on the comprehension of the selection of minimal invariant subspaces and its speed of convergence. The decomposition of the Hilbert space into a direct sum of minimal invariant subspaces generalizes the popular "Quantum Non-Demolition measurement" formalism. In our article, we have shown that the identification of the different minimal invariant states through indirect measurement allows the uniform identification of the different minimal invariant subspaces, leading to the exponential convergence result.  

We proposed a feedback control that exponentially stabilizes in mean a chosen minimal invariant subspace. This feedback is active every $N$ steps, corresponding to the length of identifiability (also possible with any $N'\geq N$). It is a promising avenue that warrants further investigation, particularly for its effectiveness in quantum error correction codes, see e.g., \cite{ahn2002continuous}. Specifically, the stabilization of subspaces 
helps to protect the code spaces from environmental perturbation.

Here we have addressed the discrete-time quantum repeated measurement models; things are slightly different in continuous time, especially concerning the identifiability condition, but follow the same line (work in progress). We have chosen to present the ideal measurement case in the paper. The generalization for imperfect measurement is straightforward, and all the proofs can be adapted. 
\section*{Acknowledgement}
    The authors would like to warmly thank Tristan Benoist for valuable discussions and improvements of the manuscript.

\bibliographystyle{plain}  
\bibliography{refs}

\end{document}